\documentclass[a4paper,10pt]{article}
\usepackage[utf8x]{inputenc}

\usepackage{epsfig}
\usepackage{algorithm}
\usepackage[noend]{algorithmic}
\usepackage{algorithmic}

\usepackage{graphicx}
\usepackage{subfigure}
\usepackage{amssymb}
\usepackage{amsmath}
\usepackage{stmaryrd}

\newtheorem{cor}{Corollary}
\newtheorem{remark}{Remark}
\newtheorem{lemma}{Lemma}

\newtheorem{proposition}{Proposition}
\newtheorem{definition}{Definition}
\newtheorem{theorem}{Theorem}

\newtheorem{proof}{Proof}

\title{Multicolored Dynamos on Toroidal Meshes}
\author{Sara Brunetti, Elena Lodi, Walter Quattrociocchi}

\begin{document}

\maketitle

\begin{abstract}
Detecting on a graph the presence of the minimum number of nodes (target set) that will be able to “activate” a prescribed number of vertices in the graph is called the {\em target set selection problem} (TSS) proposed by Kempe, Kleinberg, and Tardos. In TSS's settings, nodes have two possible states (active or non-active)  and the threshold triggering the activation of a node is given by the number of its active neighbors. 
Dealing with fault tolerance in a majority based system the two possible states are used to denote faulty or non-faulty nodes, and the threshold is given by the state of the majority of neighbors. Here, the major effort was in determining the distribution of initial faults leading the entire system to a faulty behavior. Such an activation pattern, also known as dynamic monopoly (or shortly {\em dynamo}), was introduced by Peleg in 1996. 
In this paper we extend the TSS problem's settings by representing nodes' states with a  "multicolored" set.  The extended version of the problem can be described as follows: let $G$ be a simple connected graph where every node is assigned a color from a finite ordered set $\mathcal{C}=\{1,\ldots, k\}$ of colors. At each local time step, each node can recolor itself,
depending on the local configurations, with the color held by the majority of its neighbors. Given $G$, we study the initial distributions of colors leading the 
system to a $k$ monochromatic configuration in toroidal meshes, focusing on the minimum number of initial $k$-colored nodes. We find upper and lower bounds to the size of a dynamo, and then special classes of dynamos, outlined by means of a new approach based on recoloring patterns, are characterized.
\end{abstract}

\section{Introduction}

Social behaviors often provide useful insights in determining the way problems are solved in computer science. For instance, matters of security in decision making or in collaborative content filtering/production are enforced by implementing reputation and trust strategies \cite{JosangQ09,schillo99,Sabater2007,Konig2009}.

This relation between science and the inspiration basin of social behaviors is particularly evident in the field of distributed algorithms where the  mutual dependence between local actions and global effects is fundamental, such as in the consensus or in the election problems \cite{Kossmann} \cite{Peleg96b}, \cite{Thomas79}, \cite{Burman09}, \cite{Shao09} and \cite{Garcia85}.

The {\em information diffusion} has been modeled as the spread of an information within a group through a process of social influence, where the diffusion is driven by the so called {\em influential network} \cite{Kats2005}. Such a process, which has been intensively studied under the name of {\em viral marketing} (see for instance \cite{Domingos2001}), has the goal to select an initial good set of individuals that will promote a new idea (or message) by spreading the "rumor" within the entire social network through the word-of-mouth.
The first computational study about this process \cite{Granovetter85} used the {\em linear threshold model} where the group is represented by a graph and the threshold triggering the adoption (activation) of a new idea to a node is given by the number of the active neighbors.
The impossibility of a node to return (or not) in its initial state determines the monotone (or non-monotone) behavior of the activation process. 

In a graph detecting the presence of the minimal number of nodes that will activate the whole network, namely the target set selection process (TSS), has been proved to be NP-hard through a reduction to the vertex cover\cite{Kempe03}.
In \cite{Chang09a,Ching09b} has been studied the maximum size of a {\em minimum perfect target set} under simple and strong majority thresholds -- e.g., $\lceil d(v)/2 \rceil$ and $\lceil d(v)+1/2 \rceil$ with $d(v)$ denoting the degree of a node $v$.

Other works have carried out the dynamics of majority based systems in context of fault tolerance on different networks topologies \cite{Lodi98}, \cite{Santoro03} \cite{Carvaja07}, \cite{Ching09b},
\cite{ChoudharyR09}, \cite{Mustafa04} and \cite{MustafaP01}.
In these works the major effort was in determining the distribution of initial faults leading the entire system to a faulty behavior. Such a pattern, also known as dynamic monopoly (or shortly {\em dynamo}), was introduced by Peleg \cite{Peleg96b}, and intensively studied addressing the bounds of the size of monopolies, the time to converge to a fixed point, and the topologies of the systems (see \cite{BermondBPP96}, \cite{BermondBPP03}, \cite{Bermond98}, \cite{NayakPS92}  \cite{PelegSurvey}).
In dynamic monopolies, the propagation of a faulty behavior starts by a well placed set of faulty elements and can be described as a vertex-coloring game on graphs where vertices are colored black (faulty) or white (non-faulty) and change their color at each round according to the colors of their neighbors.

Starting from the work of { \em Flocchini et al} \cite{Lodi98} we introduce an additional element to the original problem's setting: the set of the nodes' states is not limited to white or black, but vertices can assume a value from a finite and ordered set. Such a protocol, when applied on a toroidal mesh, can be described as follows: a node $x$ increments its value ($v(x)$) of one step toward the value of its neighbors if at least a couple of them has the same color greater than ($v(x)$) and a) or the remaining two have a different color in between or b) the two remaining vertices have the same color greater than $v(x)$.

As nodes are hard to persuade, due to the slow convergence process caused by the gradual convergence toward the neighbors' color, we refer to them as {\it stubborn}. 
This protocol is a clean combinatorial formulation for new contexts arising in economy, sociology, cognitive sciences, where collective decisions could be influenced by local behaviors, and a slow convergence process (due to an implicit trust strategy implemented in the protocol) would be desirable (\cite{quattrociocchi2010d,amblard01,Castellano2007}).

Our studies focus on the initial distribution of colors leading the system to a monochromatic configuration in a finite time. 
In this paper we provide a) upper and lower bounds to the size of a dynamo, and b) some special classes of dynamos by means of a new approach based on recoloring patterns. 

Due to the constant degree of nodes and regularity, toroidal meshes are efficient frameworks for these investigations. 
However, we note that results of Proposition \ref{propnew} can be easily generalized to non-constant degree graphs. 
In the current paper first we analyze the coloring properties induced by the our multi-colored protocol. 
Then bounds on the size of monotone dynamos are shown. 
We conclude the paper by characterizing special classes of dynamos and outlining the next envisioned steps of our work.   

\section{Notation and Definitions}
In this paper we study the global effects caused by the interaction among {\it stubborn} entities when disposed on a toroidal mesh. 

\begin{definition} A toroidal mesh $T :
(V,E)$ of $m \times n$ vertices is a mesh where each vertex
$v_{i,j}$ ( $0\leq i \leq m-1$ and $0\leq j \leq n-1$) is
connected to the four vertices $v_{(i-1) \mod m,j}$ , $v_{(i+1)
\mod m,j}$ , $v_{i,(j-1) \mod n}$ and $v_{i,(j+1) \mod n}$.
\end{definition}

Let $\mathcal{C}=\{1,\ldots, k\}$ be a finite ordered set of colors. A
{\em coloring} of a torus $T$ is a function $r:\; V\rightarrow \mathcal{C}$. 
If $r$ is a coloring of $T$ defined on two colors we refer to $T$ as a \textbf{bi-colored
torus}, while if $r$ is a coloring of $T$ based on more than two colors we call $T$ a \textbf{multi-colored torus}. 
$N(x)$ denotes the neighborhood of any vertex $x$ in $V$, since we are studying toroidal meshes we have that $|N(x)|=4$. 
Given a coloring $r$ of $V$, we can define the following irreversible simple majority rule \textbf{(StubSM-Protocol)}:

\begin{center}
 \begin{algorithmic}[h!]
\FOR {all $x$ $\in$ $V$} \STATE let $N(x)=\{a,b,c,d\}$   \IF {
$(r(a)=r(b)> r(x))$  $\wedge$  $((r(c) \neq r(d)) \vee (r(c) =
r(d) > r(x)))$}
        \STATE  $r(x) \gets r(x)+1$
\ENDIF \ENDFOR
\end{algorithmic}
\end{center}

For instance, let $\mathcal{C}=\{1,\ldots, 6\}$ be the finite ordered set of colors, and $T$ be the multi-colored torus shown in Figure \ref{figmulticol}. 
We represent $T$ as a matrix where the entry at the $i$th row and $j$th column is the color of $v_{i,j}$.
\begin{figure}[h!]
\begin{center}
6 4 2 4 \\
4 3 5 1 \\
6 5 2 6 \\
1 4 4 3 \\
\caption{A multi-colored torus.}
 \label{figmulticol}
\end{center}
\end{figure}

Figure \ref{figexample} illustrates the recoloring process under the 
\textbf{(StubSM-Protocol)} of the multicolored torus shown in \ref{figmulticol}.

\begin{figure}[h!]
\begin{center}
6 4 2 4   $\dashrightarrow$ 6 4 \textbf{4} 4  $\dashrightarrow$ 6 4 \textbf{4} 4     $\dashrightarrow$ 6 4 4 4  $\dashrightarrow$ 6 4 4 4  $\dashrightarrow$ 6 4 4 4       \\
4 3 5 1   $\dashrightarrow$  \textbf{5} \textbf{4} 5 \textbf{2} $\dashrightarrow$ \textbf{6} \textbf{5} 5 \textbf{3}    $\dashrightarrow$  6 5 5 \textbf{4}   $\dashrightarrow$ 6 5 5 \textbf{5}  $\dashrightarrow$ 6 5 5 \textbf{6}   \\
6 5 2 6   $\dashrightarrow$  6 5 \textbf{3} 6  $\dashrightarrow$ 6 5 \textbf{4} 6  $\dashrightarrow$    6 5 \textbf{5} 6  $\dashrightarrow$  6 5 5 6 $\dashrightarrow$ 6 5 5 6                                \\
1 4 4 3   $\dashrightarrow$  \textbf{2} 4 4 \textbf{4} $\dashrightarrow$ \textbf{3} 4 4 4   $\dashrightarrow$   \textbf{4} 4 4 4   $\dashrightarrow$ 4 4 4 4 $\dashrightarrow$ 4 4 4 4                     \\
 \caption{The coloring process of the multi-colored torus under the \textbf{(StubSM-Protocol)}.}
 \label{figexample}
\end{center}
\end{figure}

Let $r^i(x)$ be the color of $x$ after $i$ iterations of the protocol. We notice that if $r(a)=r(b)$, then $x$ recolors itself only if $r(c)\neq r(d)$ or $r(c)=r(d)>r(a)=r(b)$.
Then $r^i(x)=r(a)$ with $i=r(a)-r(x)$ unless a recoloring of its neighbors occurs. Similarly if
$r(x)<r(c)=r(d)\leq r(a)=r(b)$, then $r^i(x)=r(c)$ with
$i=r(c)-r(x)$ unless a recoloring of its neighbors occurs. The concept can be expressed more formally as follows:

\begin{lemma}
Let $x$ be in $V$, and $N(x)=\{a,b,c,d\}$ such that
$k\geq r(a)=r(b)> r(c) \neq r(d)$, and, let $i=r(a)-r(x)$. Let
$0\leq t^c_1\leq \ldots \leq t^c_i \leq i$ and $0\leq t^d_1\leq
\ldots\leq t^d_i\leq i$ be the numbers of recolorings of $r(c)$ and $r(d)$,
respectively,  at $1,\ldots,i$ time steps under the
\textbf{StubSM-Protocol}. If nodes a and b do not change color and $r^{t^c_1}(c)\neq r^{t^d_1}(d),
\ldots, r^{t^c_i}(c)\neq r^{t^d_i}(d)$, then $r^i(x)=r(a)$.
\label{lem1}
\end{lemma}
\begin{proof}
At each time step $x$ recolors itself except if $c$ and $d$ assume the same color.
\end{proof}

\begin{lemma}
Let $x$ be in $V$, and $N(x)$ = $\{a, b, c, d\}$ such that $k \geq r(a) = r(b) >
r(c) = r(d)>r(x)$. Let $0 \leq t_1^c \leq...\leq t_i^c\leq i$ and $0 \leq
t_1^d \dots \leq t_i^d\leq i$
be the numbers of recolorings of $r(c)$ and $r(d)$, respectively, at $1,
 \dots ,i$ time steps under the StubSM-Protocol. 
If nodes $a$ and $b$ do not change color and
$r^{t_1^c}(c),r^{t_1^d}(d)<r^{t_1^x}(x),\dots,r^{t_i^c}(c),r^{t_i^d}(d)<r^{t_i^x}(x)$, then $r^i(x) = r(a)$.
 \label{lem1bis}
\end{lemma}

In Figure \ref{fig:lem1bis} an example of a configuration as expressed in Lemma \ref{lem1bis} is shown.

\begin{figure}[h!]
\begin{center}
3 2 1 3 2 1 \\
2 3 2 2 3 2 \\
1 2 3 1 2 3 \\
3 2 1 3 2 1 \\
2 3 2 2 3 2 \\
1 2 3 1 2 3 \\
\caption{A multi-colored torus as expressed in Lemma \ref{lem1bis}}
 \label{fig:lem1bis}
\end{center}
\end{figure}

\begin{cor}
Let $x$ be in $V$, and $N(x)=\{a,b,c,d\}$ such that
$r(a)=r(b)=k$ and  $r(c) > r(d)$. Then
$r^i(x)=k$ with $i=k-r(x)$, if one of the two following conditions holds:
\begin{itemize}
\item[1.] $r(c)-r(d)\geq k-r(x)$;
\item[2.] $c$ and $d$ do not recolor themselves.
\end{itemize}
\label{cor1}
\end{cor}

\begin{proof}
Let $t^c_1,\ldots,t^c_i$ and $t^d_1,\ldots,t^d_i$ be recoloring of
$r(c)$ and $r(d)$ under the \textbf{StubSM-Protocol},
respectively, at time $1,\ldots,i$.  If $r(c)-r(d)\geq k-r(x)$,
then $r^{t^c_1}(c)\neq r^{t^d_1}(d), \ldots, r^{t^c_i}(c)\neq
r^{t^d_i}(d)$, and so $r^i(x)=k$. If $c$ and $d$ do not change their color, $t^c_1=\ldots=t^c_i=0$ and $t^d_1=\ldots=t^d_i=0$, and since $r(c) > r(d)$, the thesis follows by Lemma \ref{lem1}.
\end{proof}
By condition $r(c)-r(d)\geq k-r(x)$ it follows that if
$r(c)\geq r(x)$, then $r(d)\leq r(x)$, where
equalities are not contemporary true.

\begin{lemma}
Let $x$ be in $V$, and $N(x)=\{a,b,c,d\}$ such that $k\geq r(a)>r(b)> r(c) > r(d)$ or $k\geq r(a)>r(b)=r(x)=r(c) > r(d)$. If $a$ and $b$ recolor at each time step contemporary and $c$ and $d$ do not recolor themselves, then
$x$ does not recolor before $k-r(b)(>k-r(a)\geq 0)$ steps.
\label{lem2}
\end{lemma}
\begin{proof}
Vertex $x$ changes its color when at least two of its neighbors have the same color.This condition is achieved only when $a$ and $b$ assume color $k$, since the recoloring increases the color of the node $x$, and at each step the colors of $a$ and $b$ are different.
\end{proof}

We denote with $V^h$ the subset of $V$ containing all the $h$-colored vertices and the subset of $T$ of all $h$-colored vertices with $S^h$
($h\in \mathcal{C}$). 
Furthermore we denote the size of the smallest rectangle containing any $F\subseteq T$ by
and $m_F\times n_F$.

The recoloring process represents the dynamic of the system. Depending on the initial coloring of $T$, we get different dynamics. Among the possible initial configurations (that is,
assignments of colors) we are interested in those leading the system to a monochromatic configuration, namely dynamos. 
Formally,

\begin{definition}
Given an initial coloring of $T$ using colors $\mathcal{C}=\{1,\ldots, k\}$, the set $S^k$ is a \textbf{dynamo} if an all $k$-color
configuration is reached from $S^k$ in a finite number of steps
under the \textbf{StubSM-Protocol}.
\end{definition}

Besides the following definitions are needed.
\begin{definition}
Given an initial coloring of $T$ using colors in $\mathcal{C}=\{1,\ldots, k\}$, a $\textbf{h-block}$ $B^h$ is a connected subset of $T$ composed by
vertices having the same color $h$ and each node has at least
\bf{two} neighbors in $B^h$, where $h\in \mathcal{C}$.
\end{definition}

Note that vertices in $B^h$ will never change their color. For
example, $B^h$ can be a $h$-colored column (or row), any
submatrix of the adjacent rows and columns that we call a
\textbf{window}, or any $h$-colored cycle such that $v_{i,j},
v_{i,j+1}, \ldots, v_{i,j'},$ $v_{i-1,j'}, \ldots, v_{i',j'},$
$v_{i',j'-1}, \ldots v_{i',j},$ $v_{i'-1,j},\ldots v_{i,j}$ that
we call a \textbf{frame}.

\begin{definition}
A $\textbf{non-k-block}$ $NB^k$ is a connected subset of $T$ made
up of vertices of colors in $\mathcal{C}\setminus \{k\}$ each one
having at least \bf{three} neighbors in $NB^k$.
\end{definition}

This definition implies that every vertex in $NB^k$ has at most
one $k$-colored neighbor, that is vertices in $NB^k$ will never assume color $k$. For
example, two adjacent rows or columns of vertices not $k$ constitute a non-$k$-block in a toroidal mesh.

\section{Bounds on the size of a dynamo}
By means of corollary \ref{cor1}.1 we can derive the following Proposition.

\begin{proposition}
Given a coloring of the torus $T$ of size $m\times n$ such that for every vertex $x$ in $V$ and $N(x)=\{a,b,c,d\}$, $r(a)=r(b)=k$ and $r(c) \neq r(d)$ or $r(c)= r(d)> r(x)$   
then $S^k$ is a dynamo of size greater or equal to $mn/3$
\label{propnew}
\end{proposition}
\begin{proof}

By Lemma \ref{lem1} and \ref{lem1bis} immediately follows that $S^k$ is a dynamo. For each node, with a color different from $k$, it has at least two neighbors of colors $k$, and so $2$ nodes out of $5$ are $k$-colored. No conditions are imposed for the coloring of the neighborhood of $k$-colored nodes, and so $1$ node out of $5$ is $k$-colored. As a consequence,  $|S^k|\geq 2(mn-|S^k|)/5+|S^k|/5$, and hence $|S^k|\geq mn/3$.
\end{proof}

The lower bound provided by this proposition can be improved. Indeed, we are interested in determining the minimum size dynamo under the \textbf{StubSM-Protocol} for a multi-colored
toroidal mesh. This is obtained by first computing a lower bound
on the size and then an upper bound close to the lower bound.
These bounds can be derived by a reduction to the bi-colored
case (where $1$ and $2$ correspond to colors white and black, respectively).

For sake of completeness we recall here some definitions of \cite{Lodi98}.
Under the \textit{reversible simple majority rule} a white vertex turns black if at least two of its neighbors are black, otherwise the vertex does not change color, and a black vertex becomes white only if at least three of its neighbors are white; under the \textit{irreversible strong majority rule} a white vertex turns black if at least three of its neighbors are black, else the vertex does not change color, and a black vertex does not change its color. A \textit{simple} (respectively, \textit{strong}) \textit{white block} is a subset of $T$ composed of all white vertices, each of which has at least three (or respectively, two) neighbors in the block. A dynamo  is  \textit{monotone} if the set of black vertices at time $t$ is a subset of the one at time $t+1$.

We define a polynomial time transformation $\phi:
\mathcal{C}\rightarrow \mathcal{C}$ such that $\phi(h)=1$, for
$h=1,\ldots,k-1$, and $\phi(k)=2$. This transformation allows us
to map a multi-colored torus into a bi-colored torus.
Moreover under transformation $\phi$, a $non$-$k$-block
corresponds to a simple white block  and a $h$-block corresponds to a strong white block.
\begin{proposition}
A lower bound to the size of a dynamo in a bi-colored torus under
the (reversible) simple majority rule is a lower bound to the size
of a dynamo in a multi-colored torus under the
\textbf{StubSM-Protocol}. \label{prop1}
\end{proposition}

Indeed, a lower bound consists in the smallest size of $S^k$
such that no $non$-$k$-blocks can arise in the multi-colored problem, and in the smallest size of $S^2$ (initial set of black vertices) such that no simple white blocks
can arise in the bi-colored setting.

Because of the correspondence between a $non$-$k$-block and a
simple white block the claim follows. Therefore we derive (see
Theorem 9 of \cite{Lodi98} for simple monotone dynamos):
\begin{theorem}
Let $S^k$ be a dynamo for a colored toroidal mesh of size $m
\times n$ under the
\textbf{StubSM-Protocol}. We have
\begin{itemize}
\item (i) $m_{S^k}\geq m-1,\; n_{S^k}\geq n-1$ \item (ii)
$|S^k|\geq m+n-2$.
\end{itemize}
\label{t9}
\end{theorem}
\begin{figure}[h!]
\begin{center}
2 2 1 1 1 1 1 1\\
2 2 1 1 1 1 1 1\\
1 1 2 2 1 1 1 1\\
1 1 2 2 1 2 2 1\\
1 1 1 1 1 2 2 1\\
1 1 1 1 1 1 1 1
\caption{A monotone dynamo of size $m+n-2$.}
 \label{figB1}
\end{center}
\end{figure}
Figure \ref{figB1} illustrates a monotone dynamo (under the simple reverse majority rule) in a bi-colored torus. We notice that by mapping color $2$ with $k$, it is possible to find out an assignment of colors of $\mathcal{C}\setminus\{k\}$ for $1$-colored vertices such that the obtained multi-colored torus leads to a monochromatic configuration under the \textbf{StubSM-Protocol}.

\begin{proposition}
An upper bound to the size of a dynamo in a bi-colored torus under the (irreversible) strong majority rule is an upper bound to the size of a dynamo in a multi-colored torus under the
\textbf{StubSM-Protocol}. \label{prop2}
\end{proposition}

Indeed in order to establish an upper bound to the size of $S^k$,
no $h$-blocks have to appear with $h=1,\ldots, k-1$ and successive derived $k$-colored
sets of vertices have to contain the set $V$ of all vertices
at the end of the process. 

We have that: a) strong white blocks correspond to $h$-blocks; b) irreversible
strong majority rule is more restrictive than
\textbf{StubSM-Protocol}: hence, under the irreversible strong
majority rule, a vertex recolors itself if there are three
vertices in its neighborhood having the same color, whereas under
the \textbf{StubSM-Protocol} are needed two neighbors with the same color (and the reamaining ones with different colors in between).
Hence, in order to obtain an upper bound to the size of $S^2$,
no strong white blocks have to arise and successive derived black sets of vertices 
have to contain the set $V$ of all the vertices  at the end of the process, in the bi-colored
problem, because of a) and b) the claim follows. 
Therefore we get (see Theorem 8 of \cite{Lodi98}):
\begin{theorem}
Let $S^k$ be a dynamo for a colored toroidal mesh of size $m
\times n$ under the
\textbf{StubSM-Protocol}. Then  $|S^k|\geq \lceil{m/3}\rceil (n+1)$. \label{t8}
\end{theorem}
\begin{figure}[h!]
\begin{center}
2 1 1 1 1 1 1 1\\
1 2 1 2 1 2 1 2\\
2 1 2 1 2 1 2 1\\
2 1 1 1 1 1 1 1\\
1 2 1 2 1 2 1 2\\
2 1 2 1 2 1 2 1\\
2 1 1 1 1 1 1 1\\
1 2 1 2 1 2 1 2\\
2 1 2 1 2 1 2 1
\caption{A strong irreversible dynamo of size $\lceil m/3\rceil (n+1)$.}
 \label{figB2}
 \end{center}
 \end{figure}
Figure \ref{figB2} illustrates a strong irreversible dynamo of size $\lceil m/3\rceil (n+1)$.
Note that for every assignment of colors of $\mathcal{C}\setminus\{k\}$, $S^k$ illustrated in Figure \ref{figB2} is a dynamo.
\section{A minimum size dynamo}
Theorem \ref{t8} establishes an upper bound far from the lower
bound determined in Theorem \ref{t9}. 
In this section a minimum size dynamo is derived. 
If we choose $S^k$ as made up of the first row and column in the torus, then $|S^k|=m+n-1$ that is close to the lower bound in Theorem \ref{t9}.

\begin{lemma}
Let $S^k$ be a dynamo. Then, $T-S^k$ does not contain any
$h$-block, with $h\in \mathcal{C}\setminus \{k\}$. \label{l2}
\end{lemma}
Our choice for  $S^k$ implies that no $h$-colored column or $h$-colored row, but an
$h$-colored window or an $h$-colored frame can arise, with $h\in
\mathcal{C}\setminus \{k\}$. As a consequence we require that:\\

\noindent for every $2\times 2$ window in $T$, $r(v_{i,j})\neq
r(v_{i+1,j+1})$ and $r(v_{i,j+1})\neq r(v_{i+1,j})$; otherwise
$r(v_{i,j})= r(v_{i+1,j+1})=k$ ($r(v_{i,j+1})=r(v_{i+1,j})=k$).\\

This requirement does not forbid the appearance of a $h$-colored window during the recoloring
process as shown Figure \ref{fig1}.

\begin{figure}[h!]
\begin{center}
8 8 8 8 8 8   $\dashrightarrow$   8 8 8 8 8 8\\
8 7 7 3 4 6   $\dashrightarrow$   8 8 8 6 8 8\\
8 1 4 3 7 5   $\dashrightarrow$   8 5 \textbf{4} \textbf{4} 7 8\\
8 1 4 3 2 4   $\dashrightarrow$   8 3 \textbf{4} \textbf{4} 2 8\\
8 1 4 3 2 4   $\dashrightarrow$   8 6 4 3 4 8
 \caption{An Example in which a $4$-block emerges after five steps.}
 \label{fig1}
\end{center}
\end{figure}

Therefore we focus on the recoloring process providing a
certain $h$-block in order to avoid it.

\begin{remark}
Let $a_W,b_W,c_W,d_W$ be the vertices of any $2\times 2$ window $W$
with $r(a_W)\leq r(b_W) \leq r(c_W)\leq r(d_W)$ and $i$ and $j$ be
the number of recoloring of $a_W$ and $d_W$, respectively, under
the  \textbf{StubSM-Protocol}. No $h$-block can appear into $W$ during the recoloring process if
$i-j<r(d_W)-r(a_W)$, with $(k>)h\geq r(d_W)$. 

In the example illustrated in Figure \ref{fig1} $r(a_W)=r(v_{3,3})=3=r(b_W)=r(v_{2,3}),
\;r(c_W)=r(v_{3,2})=4=r(d_W)=r(v_{2,2})$: as a consequence of the
recoloring of $v_{1,3}$, node $v_{2,3}$ assumes color $4$, and
hence at the next (fourth) step $v_{3,3}$ recolors with $4$ by
causing the formation of a $4$-block. 
Note that $i-j=1-0=4-3$.
\end{remark}

%
%
Remark suggests that a dynamo can be outlined by an assignment of the initial distribution of colors which takes into account the recoloring pattern due to the $S^k$ considered.

Let us add a fictitious color $\infty$ to the finite set $\mathcal{C}$ of colors. This color is the greatest color, but two colors $\infty$ are not comparable.  

\begin{definition}
A \textbf{NordWest-window} $T^{NW}(i^*,j^*)$ of a $m\times n$ toroidal mesh $T$ in $(i^*,j^*)$ is the submesh of $T$ having  vertices $v_{i,j}$, $0\leq i<i^*<m$ and $0\leq j<j^*<n$ augmented by a $i^*$-th row and by a $j^*$-th column of vertices colored by $\infty$.
\end{definition}

\begin{definition}
A \textbf{NordEast-window} $T^{NE}(i^*,j^*)$ of a $m\times n$ toroidal mesh $T$ in $(i^*,j^*)$ is the submesh of $T$ having  vertices $v_{i,j}$, $0\leq i<i^*<m$ and $0<j^*< j< n$ augmented by a $i^*$-th row and by a $j^*$-th column of vertices colored by $\infty$.
\end{definition}

\begin{definition}
A \textbf{SouthWest-window} $T^{SW}(i^*,j^*)$ of a $m\times n$ toroidal mesh $T$ in $(i^*,j^*)$ is the submesh of $T$ having  vertices $v_{i,j}$, $0<i^*<i < m$ and $0\leq j<j^*< n$ augmented by a $i^*$-th row and by a $j^*$-th column of vertices colored by $\infty$.
\end{definition}

\begin{definition}
A \textbf{SouthEast-window} $T^{SE}(i^*,j^*)$ of a $m\times n$ toroidal mesh $T$ in $(i^*,j^*)$ is the submesh of $T$ having  vertices $v_{i,j}$, $0<i^*<i < m$ and $0<j^*< j < n$ augmented by a $i^*$-th row and by a $j^*$-th column of vertices colored by $\infty$.
\end{definition}

\begin{lemma}
Let $T^{NW}(i^*,j^*)$ be the NordWest-window of a $m\times n$ toroidal mesh $T$ such that
\begin{itemize}
\item $v_{i,0}=v_{0,j}=k$, for $i=0,\ldots,i^*-1$ and $j=1,\ldots, j^*-1$;
\item $v_{i,1}\geq \ldots \geq v_{i,j^*-1}$, for $0<i<i^*$;
\item $v_{1,j}\geq \ldots \geq v_{i^*-1,j}$, for $0<j<j^*$;
\item $v_{i,j} >v_{i+1,j-1}$ for all $i,j$ such that $i+j-1=l$, for $3\leq l < i^*+j^*-3$.
\end{itemize}
Then, all the vertices of $T^{NW}(i^*,j^*)\cap T$ recolor by $k$ after
$$ M(i^*-1,j^*-1)= max \left\{ \begin{array}{c} M(i^*-1,j^*-2) \\  M(i^*-2,j^*-1))\end{array} +k-r(v_{i^*-1,j^*-1})\right.$$ steps, with $M(0,j)=M(i,0)=0$, for $i=0,\ldots,i^*-1$ and $j=1,\ldots, j^*-1$.
\label{lemNW}
\end{lemma}
\begin{proof}
Let $M(i,j)$ denote the number of steps needed for vertex $v_{i,j}$ to reach $k$. 
By the first condition we get that $M(0,j)=M(i,0)=0$, for $i=0,\ldots,i^*-1$ and $j=1,\ldots, j^*-1$.
Given the initial configuration of $T^{NW}(i^*,j^*)$, in the first round only the node $v_{1,1}$ recolors itself, since $r(v_{0,1})=r(v_{1,0})=k$ and $r(v_{1,2})>r(v_{2,1})$, whereas all the other vertices have neighbors with different colors or two neighbors of color $\infty$. In the second round, the recoloring of $v_{1,1}$ changes the chromatic configuration of the neighborhoods of $v_{1,2}$ and $v_{2,1}$. 
By Lemma \ref{lem2} with $x=v_{1,2}$ and $r(a)=r(v_{0,2})=k$ ($x=v_{2,1}$ and $r(a)=r(v_{2,0})=k$), $v_{1,2}$ does not change color until node $v(1,1)$ reaches $k$. This pattern happens in $k-r(v_{1,1})$ steps, hence $M(1,1)=k-r(v_{1,1})$ verifies the relation. 
At the $k-r(v_{1,1})+1$th step,  $v_{1,2}$ and $v_{2,1}$ recolor themselves while all the other vertices do not change. The recoloring of $v_{1,2}$ and $v_{2,1}$ changes the colors of the neighborhoods of $v_{1,3},\; v_{2,2}$ and $v_{3,1}$. 
By Lemma \ref{lem2} $v_{2,2}$ does not advance and nodes $v_{1,3}$ and $v_{3,1}$ do the same. 
The node $v_{1,2}$ recolors after additional $k-r(v_{1,2})$, that is $M(1,2)=max(M(1,1),0)+k-r(v_{1,2})= k-r(v_{1,1})+ k-r(v_{1,2})$, and $M(2,1)=max(0,M(1,1))+k-r(v_{2,1})= k-r(v_{1,1})+ k-r(v_{2,1})$. Therefore, first node $v_{1,3}$ and, then $v_{2,2}$ and $v_{3,1}$, start recoloring, being $r(v_{1,2})>r(v_{2,1})$. By the same considerations mentioned before, we can conclude that $v_{1,3}$ and, then $v_{2,2}$ and $v_{3,1}$ become $k$-colored after $M(1,3)=max(M(1,2),0)+k-r(v_{1,3})$, $M(2,2)=max(M(2,1),M(1,2))+k-r(v_{2,2})$ and $M(3,1)=max(0, M(2,1))+k-r(v_{3,1})$ steps respectively. At the end of the process all vertices in $T^{NW}(i^*,j^*)\cap T$ are $k$-colored. 
\end{proof}

An analogous lemmas can be stated for the NordEast-window, SouthWest-window and SouthEast-window of $T$.
\begin{lemma}
Let $T^{NE}(i^*,j^*)$ be the NordEast-window of a $m\times n$ toroidal mesh $T$ such that
\begin{itemize}
\item $v_{i,n-1}=v_{0,j}=k$, for $i=0,\ldots,i^*-1$ and $j=j^*+1,\ldots, n-1$;
\item $v_{i,j^*+1}\leq \ldots \leq v_{i,n-1}$, for $0<i<i^*$;
\item $v_{1,j}\geq \ldots \geq v_{i^*-1,j}$, for $j^*<j<n$;
\item $v_{i,j} >v_{i+1,j+1}$ for all $i,j$ such that $n-j+i=l$, for $3\leq l < n-j^*+i^*-3$.
\end{itemize}
Then, all the vertices of $T^{NE}(i^*,j^*)\cap T$ recolor by $k$ after
$$ M(i^*-1,j^*+1)= max \left\{ \begin{array}{c} M(i^*-1,j^*+2) \\  M(i^*-2,j^*+1))\end{array} +k-r(v_{i^*-1,j^*+1})\right.$$ steps, with $M(0,j)=M(i,n-1)=0$, for $i=0,\ldots,i^*-1$ and $j=j^*+1,\ldots, n-1$.
\label{lemNE}
\end{lemma}

\begin{lemma}
Let $T^{SW}(i^*,j^*)$ be the SouthWest-window of a $m\times n$ toroidal mesh $T$ such that
\begin{itemize}
\item $v_{i,0}=v_{m-1,j}=k$, for $i=i^*+1,\ldots,m-1$ and $j=1,\ldots, j^*-1$;
\item $v_{i,1}\geq \ldots \geq v_{i,j^*-1}$, for $i^*+1<i<m$;
\item $v_{1,j}\leq \ldots \leq v_{i^*-1,j}$, for $0<j<j^*$;
\item $v_{i,j} <v_{i+1,j-1}$ for all $i,j$ such that $i+j-1=l$, for $i^++2\leq l < m+j^*-4$.
\end{itemize}
Then, all the vertices of $T^{NW}(i^*,j^*)\cap T$ recolor by $k$ after
$$ M(i^*+1,j^*-1)= max \left\{ \begin{array}{c} M(i^*+1,j^*-2) \\  M(i^*+2,j^*-1))\end{array} +k-r(v_{i^*+1,j^*-1})\right.$$ steps, with $M(m-1,j)=M(i,0)=0$, for $i=i^*+1,\ldots,m-1$ and $j=1,\ldots, j^*-1$.
\label{lemSW}
\end{lemma}

\begin{lemma}
Let $T^{SE}(i^*,j^*)$ be the SouthEast-window of a $m\times n$ toroidal mesh $T$ such that
\begin{itemize}
\item $v_{i,n-1}=v_{m-1,j}=k$, for $i=i^*+1,\ldots,m-1$ and $j=j^*+1,\ldots, n-1$;
\item $v_{i,j^*+1}\leq \ldots \leq v_{i,n}$, for $0<i<i^*$;
\item $v_{1,j}\leq \ldots \leq v_{i^*-1,j}$, for $j^*<j<n$;
\item $v_{i,j} <v_{i+1,j+1}$ for all $i,j$ such that $n-j+i=l$, for $i^*+3\leq l < n-j^*+m-3$.
\end{itemize}
Then, all the vertices of $T^{SE}(i^*,j^*)\cap T$ recolor by $k$ after
$$ M(i^*+1,j^*+1)= max \left\{ \begin{array}{c} M(i^*+1,j^*+2) \\  M(i^*+2,j^*+1))\end{array} +k-r(v_{i^*+1,j^*+1})\right.$$ steps, with $M(m-1,j)=M(i,n-1)=0$, for $i=i^*+1,\ldots,m-1$ and $j=j^*+1,\ldots, n-1$.
\label{lemSE}
\end{lemma}

Let $r_i=r(v_{i,1})=r(v_{i,2})=\ldots =r(v_{i,n-1})$, for $i=0,\ldots, m-1$, and $c_0=r(v_{0,0})=r(v_{1,0})=\ldots =r(v_{m-1,0})$.
\begin{theorem}
Given a coloring $r$ of the vertices of a $m\times n$ toroidal mesh $T$, let $S^k$ be constituted by the first row and column, i.e. $r_0=k$ and $c_0=k$; let $r_i=r_{m-i}$, $r_i>r_{i+1}$, and $r_{m-i}>r_{m-i-1}$ for $i=1,\ldots,\lceil m/2\rceil-2$;\begin{itemize} \item if $m$ is even: let  1) $r_{m/2-1}>r_{m/2}, r_{m/2+1}$, 2) $r_{m/2+1}> r_{m/2}$, 3) $r_{m/2-1}+r_{m/2}<2 r_{m/2+1}$; \item if $m$ is odd: let 1) $r_{\lceil m/2\rceil-1}>r_{\lceil m/2\rceil}$, 2) $k+r_{\lceil m/2\rceil}<2 r_{\lceil m/2\rceil-1}$;  \end{itemize} where $k-r_{\lceil m/2 \rceil-1}\geq \lceil m/2 \rceil-1$.
Then $S^k$ is a dynamo.
\end{theorem}
\begin{proof}
By Lemmas \ref{lemNW}-\ref{lemSE} all the vertices of $T^{NW}(\lceil m/2\rceil-1,\lceil n/2 \rceil)$,
$T^{NE}(\lceil m/2\rceil-1,\lceil n/2 \rceil-1)$, $T^{SW}(\lceil m/2\rceil+1,\lceil n/2 \rceil)$, $T^{SE}(\lceil m/2\rceil+1,\lceil n/2 \rceil-1)$ recolors by $k$
at the end of the recoloring process. Consider now rows $\lceil m/2\rceil-1$, $\lceil m/2\rceil$, $\lceil m/2\rceil+1$.  \begin{itemize} \item Let $m$ be even:  since  $r_{m/2-1}>r_{m/2}, r_{m/2+1}$ (condition 1), we have that $v_{m/2,1}$ ($v_{m/2,n-1}$) starts to recolor itself after that $v_{m/2-1,1}$ ($v_{m/2-1,n-1}$) recolored with $k$.
We are going to show that vertices of rows $m/2-1$ and $m/2+1$ recolor with $k$ by Corollary \ref{cor1}.2, because every vertex of the $m/2$th-row starts to change recolor only when its neighbors on the rows $m/2-1$ and $m/2+1$ are $k$-colored. We prove that the color assumed by $v_{m/2,1}$ during recoloring is different from the pattern of $v_{m/2+1,2}$. 
The number of recolorings of $v_{m/2,1}$ needed to obtain the same color as $r(v_{m/2+1,2})$ is $ r_{m/2+1}- r_{m/2}$. Since $r(v_{m/2+1,2})>r(v_{m/2,1})$ (condition 2)
and $v_{m/2+1,2}$ starts to recolor after that $v_{m/2+1,1}$ recolored with $k$, if this happen before $ r_{m/2+1}- r_{m/2}$ steps, the colors of $v_{m/2,1}$ and $v_{m/2+1,2}$ at each time step are different. We have that $k-(r_{m/2+1}+k-r_{m/2-1})=r_{m/2-1}-r_{m/2+1}$ which is less than $ r_{m/2+1}- r_{m/2}$ by condition 3. Therefore, in a first phase $v_{m/2+1,1}$ recolors by $k$, then $v_{m/2,2}$ starts to recolor after that $v_{m/2+1,2}$ recolored with $k$ and in a third moment, by Lemma \ref{lem1}, $v_{m/2,2}$ will become $k$-colored. 
Finally this can be proved for all the vertices on the $m/2$th-row, in order that no $h$-block could arise with $h\neq k$, and at the end of the process all the vertices will be $k$-colored.
\item Let $m$ be odd:  $v_{\lceil m/2\rceil-1,1}$ and $v_{\lceil m/2\rceil,1}$ start to recolor at the same time. Although  $r_{\lceil m/2\rceil-1}>r_{\lceil m/2\rceil}$ (condition 1), we are going to show that the color of
 $v_{\lceil m/2\rceil,1}$ assumed during recolorings is different from the color of  $v_{\lceil m/2\rceil-1,2}$. Indeed
$r_{\lceil m/2\rceil}+k-r_{\lceil m/2\rceil-1}<r_{\lceil m/2\rceil-1}$ by condition 2. Therefore every vertex of row ${\lceil m/2\rceil}$ starts to recolor when it has two neighbors of color $k$, and remaining neighbors of different colors. By Lemma \ref{lem1} every vertex will recolor by $k$, and no $h$-block will appear with $h\neq k$.

\end{itemize}  \end{proof}

\begin{figure}[h!]
\begin{center}
6 6 6 6 6    \hspace{.2in}   6 6 6 6 6\\
6 5 5 5 5    \hspace{.2in}   6 5 5 5 5\\
6 4 4 4 4    \hspace{.2in}   6 4 4 4 4\\
6 1 1 1 1    \hspace{.2in}   6 1 1 1 1\\
6 3 3 3 3    \hspace{.2in}   6 5 5 5 5\\
6 5 5 5 5    \hspace{.82in}
 \caption{An example of a dynamo for $m$ even and for $m$ odd.}
 \label{figF}
\end{center}
\end{figure}

\section{Conclusions}

In this paper we introduced the {\em multicolored dynamos}, a new problem that is an extension of the original {\em target set selection} (TSS) problem.
In our settings the set of the nodes' states is not limited to white or black, such as for the dynamic monopolies, but vertices can assume values from a finite and ordered set.
This protocol finds application in contexts where the collective decisions can be influenced by malicious behaviors, e.g. partial copies of corrupted data or faulty sensors, and a slow convergence process (due to an implicit trust strategy implemented in the protocol) would be desirable. In this work we characterized the nodes' coloring patterns in terms of neighbors' influence, that is a function of the nodes' degrees.

At the end of this work, there are some case studies and several interesting questions that still remain open. For instance, other kinds of topologies could be considered under the
\textbf{SMP-Protocol} such as scale-free networks or random graphs to have a comparitive analysis with respect to other algorithmic models of social influence and viral marketing on social network, e.g. the bounded confidence model \cite{amblard01}.
Furthermore, considering the growing attention to the dynamic aspects of social networks such a protocol should be studied on graphs where the availability of links and nodes is subject to change during time \cite{CFQS2010a}. 
This statement leads to a different definition of majority and the deriving propagation patterns should be investigated and characterized. Moreover, instead of studying initial configurations leading to a monochromatic configuration, the problem of determining initial configurations avoiding the converge of the whole system toward a monochromatic fixed point could be investigated.

\bibliographystyle{plain}
\bibliography{biblio}

\end{document}